\documentclass[10pt, conference, letterpaper]{IEEEtran}\IEEEoverridecommandlockouts 
\usepackage{url}
\usepackage{amssymb}
\usepackage{amsmath}
\usepackage{amsfonts}
\usepackage{graphicx}
\usepackage{algorithm}
\usepackage{algorithmic}
\usepackage{cite}
\usepackage{epstopdf}
\usepackage{stfloats}
\usepackage{color}
\usepackage{amsthm}
\newtheorem{corollary}{\textbf{Corollary}}

\newtheorem{proposition}{\textbf{Proposition}}
\newtheorem{remark}{\textbf{Remark}}

\newcommand{\diag}{\mathop{\mathrm{diag}}}
\newcommand{\Tr}{\mathrm{Tr}}
\newcommand{\SNR}{\mathrm{SNR}}

\newcommand{\true}{honest }
\newcommand{\True}{Honest }

\begin{document}
\bibliographystyle{IEEEtran}

\title{Impact of Channel State Misreporting on Multi-user Massive MIMO Scheduling Performance}
\author{Zhanzhan Zhang$^{\dag}$, Yin Sun$^{*}$, Ashutosh Sabharwal$^{\ddag}$, \IEEEmembership{Fellow,~IEEE}, and Zhiyong Chen$^{\dag}$\\
$^{\dag}$Cooperative Medianet Innovation Center, Shanghai Jiao Tong University, P. R. China\\
$^{*}$Dept. of ECE, Auburn University, AL, USA\\
$^{\ddag}$Dept. of ECE, Rice University, TX, USA\\
Email: \{mingzhanzhang, zhiyongchen\}@sjtu.edu.cn, yzs0078@auburn.edu, ashu@rice.edu
\thanks{Yin Sun was partially supported by the Office of Naval Research under Grant N00014-17-1-2417.

Ashutosh Sabharwal was partially supported by NSF grants CNS-1518916 and CNS-1314822.

Zhiyong Chen was partially supported by the National Natural Science Foundation of China (Grant No. 61671291, 61528103, and 61521062), Huawei HIRP Project under Grants YB2015040062 and STCSM15DZ2270400.}}
\maketitle

\begin{abstract}
The robustness of system throughput with scheduling is a critical issue. In this paper, we analyze the sensitivity of multi-user scheduling performance to channel misreporting in systems with massive antennas. The main result is that for the round-robin scheduler combined with max-min power control, the channel magnitude misreporting is harmful to the scheduling performance and has a different impact from the purely physical layer analysis. Specifically, for the homogeneous users that have equal average signal-to-noise ratios (SNRs), underreporting is harmful, while overreporting is beneficial to others. In underreporting, the asymptotic rate loss on others is derived, which is tight when the number of antennas is huge. One interesting observation in our research is that the rate loss ``periodically'' increases and decreases as the number of misreporters grows. For the heterogeneous users that have various SNRs, both underreporting and overreporting can degrade the scheduler performance. We observe that strong misreporting changes the user grouping decision and hence greatly decreases some users' rates regardless of others gaining rate improvements, while with carefully designed weak misreporting, the scheduling decision keeps fixed and the rate loss on others is shown to grow nearly linearly with the number of misreporters.
\end{abstract}

\section{Introduction}
Multi-user multiple-input multiple-output (MU-MIMO) continues to be a key technology for current and future networks \cite{MUMIMO13inLTE}; notably massive MIMO systems are aiming to support tens of users in one time-frequency resource block in future generations systems \cite{massiveMIMOnextg} to significantly increase spectral efficiencies. Joint transmissions in one time-frequency resource block highlight the importance of scheduling, since in a joint transmission, the downlink packets to different users share the base-station (BS) transmission power. The transmission power allocation among users is often performed based on multiple factors, like path loss coupled with rate and reliability demands. And channel knowledge is crucial for scheduling and hence most systems have extensive protocol support for mobile clients to report channel state information (CSI) to the base-station.

In this paper, we study the impact of channel state \emph{misreporting} on scheduler performance. Channel state misreporting can be either intentional (e.g.,\ malicious attacker or greedy users) or unintentional (e.g.,\ software errors); an example of coordinated multi-IoT device attack was reported recently~\cite{IoT-attack}.
It is clear that the scheduler performance will depend on the accuracy of channel state information, since the multi-user capacity region (implicitly or explicitly) computed by the scheduler depends on the reported channel state information. We note that in frequency-division duplex (FDD), which is the more prevalent form of networks, base-stations \emph{have to} rely on mobiles to report their measured downlink channel conditions, thereby creating a possibility of potential misreporting.

Note that the MIMO channel misreporting can occur in multiple ways \cite{tung2014vulnerability}: (i)~misreporting the channel direction, and (ii)~misreporting of the channel magnitude, which further includes \emph{underreporting} and \emph{overreporting} \cite{kim2014falseCSIreport}. In this paper, we study the second type of misreporting, i.e.,\ misreporting of the channel magnitudes or downlink signal-to-noise ratio (SNR) measured by the mobile. While both forms of misreporting are crucial, we will focus on channel SNR misreporting as a first concrete step in understanding the impact of misreporting on scheduled system performance. Additionally, misreporting the channel direction is easier to detect by some angle-of-arrival (AoA) based techniques, see e.g., \cite{xiong2013securearray, li2017analog}, since the AoAs between downlink and uplink in FDD massive MIMO are correlated \cite{xing15AoA}. In contrast, channel magnitude misreporting is difficult to detect~\cite{kim2014falseCSIreport}, since the channel magnitudes change naturally in wireless links due to large-scale and small-scale channel fading. 

In multi-user downlink transmissions, the scheduler has to make decisions on \emph{which} users are grouped together in each resource block, and \emph{how} the power is divided among the users of each group. These decisions depend on the type of downlink precoding methods (e.g., zero-forcing (ZF) or conjugate beamforming) used in the physical layer, as different methods lead to different achievable rates. Since each resource block can accommodate tens of users, the number of misreporting users also impacts the scheduling performance~\cite{IoT-attack}. Hence, the actual rate loss depends on user grouping, power allocation, multi-user beamforming, and the number of misreporters, which make the overall analysis fairly challenging.

In this work, we consider the multi-user massive MIMO network in a single cell operating in the FDD mode, where the BS is equipped with a large number of antennas. For tractability, we focus on the round-robin (RR) scheduling that uses the channel magnitude based user grouping, combined with max-min fairness power allocation and zero-forcing precoding. As a point of comparison, we also analyze the performances of two additional user grouping methods: semi-orthogonal user selection (SUS) \cite{SUSZF} and random selection. Both the channel underreporting and overreporting are analyzed. To the best of the authors' knowledge, this is the first work to investigate the impact of channel misreporting on multi-user scheduling in massive MIMO systems.

The main contributions of the paper are as follows:
\begin{itemize}
  \item We first analyze a homogeneous case, where all users have equal average SNRs. We find that the channel magnitude underreporting by a few users can decrease other users' data rates, while channel overreporting is beneficial to others. In underreporting, we derive the asymptotic closed-form expression of the rate loss on \true users, which is shown to be a quite accurate approximation for a moderate number of antennas.


  \item 
      In underreporting, we discover that there is a ``periodicity" in the rate loss as a function of the number of misreporting users, with the ``period" being equal to the number of users per resource block. That is, the rate loss function ``periodically'' increases and decreases with the number of misreporting users. The reason is that the rate loss mainly comes from the \emph{infected} resource blocks that support both the misreporting and \true users, and the number of \true users in infected resource blocks varies periodically. 
  The seemingly counter-intuitive result demonstrates how the scheduler performance loss analysis can yield different results compared to purely physical layer analyses without scheduling \cite{Muk10ICASSP, tung2014vulnerability} where the performance loss increases monotonically with more misreporters.

  \item For the heterogeneous case where users have various average SNRs, both underreporting and overreporting can decrease the scheduler performance, which is different from the homogeneous case.
      We propose two efficient misreporting strategies: user grouping changed misreporting and user grouping unchanged misreporting. Specifically, the former misreporting strategy does its best to harm a part of the \true users by changing the user grouping, yet it also benefits the other \true users.
      In contrast, the latter misreporting strategy keeps the user grouping fixed by carefully designing the misreporting levels. With the latter strategy, no \true user gets rate improvement, and the average rate loss is shown to grow nearly linearly with the number of misreporters.
\end{itemize}


\vspace{0mm}

Channel misreporting on massive MIMO scheduling performance differs significantly from the prior works. The works in \cite{tung2014vulnerability} and \cite{Muk10ICASSP} revealed the new threats of channel misreporting on precoding process, without considering user grouping and scheduling. However, one misreporter may fail to harm other users under user scheduling, since it may not be selected or can be scheduled to a different resource block. In addition, the similar issue on scheduling was addressed in \cite{kim2014falseCSIreport} and \cite{Racic10TMC} for the current 3G and LTE systems, where only one user is selected in one resource block. Thus, the adversaries can harm other users and cause long inter-packet delays by simply occupying consecutive resource blocks. However, in multi-user case, it's more difficult to occupy many consecutive resource blocks, since each resource block serves multiple users. In contrast, the threats of channel misreporting on multi-user scheduling consist of two new features that have not been studied before: (i) channel misreporting leads to unfair power allocation among honest and misreporting users that are both served in the infected resource block, and (ii) the changes in user grouping decisions can greatly affect the throughput of the users in different resource blocks.



\textit{\textbf{Notations:}} Boldface uppercase and boldface lowercase letters  denote matrices and column vectors, respectively. $\mathbb{E}\{\cdot\}$, $\|{\cdot}\|$, $\Tr(\cdot)$, $(\cdot)^H$, $(\cdot)^T$, $(\cdot)^{-1}$ stand for the expectation, Euclidean norm, the trace of a square matrix, the conjugate transpose, the transpose and the inverse of a matrix, respectively. $\mathcal{CN}(\mathbf{x}, \mathbf{\Sigma})$ represents the distribution of a circularly symmetric complex Gaussian vector with mean vector  $\mathbf{x}$ and covariance matrix $\mathbf{\Sigma}$.  $\textbf{I}_M$ denotes an $M{\times} M$ identity matrix.

\vspace{-3pt}

\section{Model}
\vspace{-2pt}
\begin{figure}
\setlength{\abovecaptionskip}{-1mm}
  \centering
  \includegraphics[width=2.7in]{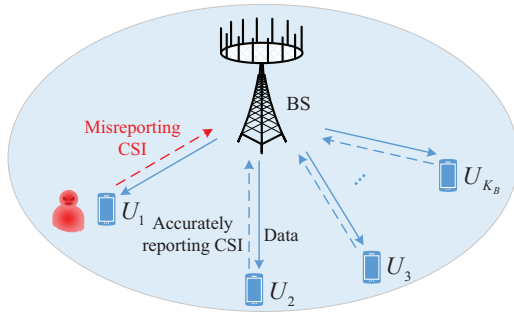}
  \caption{Multi-user massive MIMO downlink in FDD networks.}\label{systemModel}
  \vspace{-7mm}
\end{figure}

\vspace{-3pt}
\subsection{System Model}
As illustrated in Fig. \ref{systemModel}, we consider a multi-user massive MIMO downlink network which operates in FDD mode. The BS is employed with a massive antenna array which has $M$ ($M{\gg} 1$) antennas, and there are total $K$ users with one antenna at each user. For downlink transmissions, the BS requires users to report the channel estimation. The downlink channel estimation and uplink feedback phases are assumed to be error-free, which allows us to focus on the impact of error caused by misreporting. The round-robin scheduler is adopted to schedule users that are grouped based on their channel magnitudes. Denote by $K_{B}$ the number of users that the BS can beamform to in one time-frequency resource block, and we assume $K=T K_{B}$ without loss of generality, where $T$ consecutive resource blocks constitute one RR scheduling period, during which all users are served once. And denote by $K_M$ the number of misreporting users. Let $P$ denote the power constraint transmitted by the BS in each resource block. Moreover, taking fairness into consideration,  the BS is assumed to use the max-min power control that ensures equal received SNR for the users in each resource block.

In the downlink data transmissions, the BS broadcasts the signal to the selected $K_{B}$ users through ZF precoding. Let $\mathbf{s}\in \mathbb{C}^{K_{B}{\times} 1}$ be the Gaussian source data vector that has unit variance entries. The channel from the BS to the $k$-th user is denoted by $\mathbf{g}_k\in\mathbb{C}^{1{\times} M}$ and is modeled as $\mathbf{g}_k=\sqrt{\beta_k}\mathbf{h}_k$, where $\beta_k$ represents the large-scale fading, and $\mathbf{h}_k$ denotes the small-scale fading. The channel vector $\mathbf{g}_k$ is assumed to obey the frequency-flat independent identically distributed (i.i.d.) Rayleigh fading, thus we have $\mathbf{g}_k \sim \mathcal{CN}(\mathbf{0}, \beta_k \mathbf{I}_M)$. We define $\mathbf{G} \triangleq [\mathbf{g}_1^T, \cdots, \mathbf{g}_{K_B}^T]^T$, and define $\mathbf{H} \triangleq [\mathbf{h}_1^T, \cdots, \mathbf{h}_{K_B}^T]^T$, hence the channel matrix $\mathbf{G}$ can be expressed as $\mathbf{G}=\mathbf{B}^{1/2} \mathbf{H}$, where $\mathbf{B}$ is a diagonal matrix with $\{\beta_k\}$ on its diagonal.
Then, the precoded signal vector is given by
\vspace{-4pt}
\begin{equation}\label{precodedSignal}
\mathbf{x} = \mathbf{WDPs},
\vspace{-4pt}
\end{equation}
where $\mathbf{P}$ is a diagonal matrix representing the power allocation among the $K_{B}$ users with $\{  \sqrt{P_k} \}$ on its diagonal, and $\mathbf{P}$ satisfies the power constraint $\Tr(\mathbf{P}^2)=P$. The matrix product $\mathbf{WD}$ is the beamforming matrix, and $\mathbf{W}$ represents the pseudoinverse of the channel matrix, i.e., $\mathbf{W}=\mathbf{G}^H(\mathbf{GG}^H)^{-1}$, and $\mathbf{D}$ is a diagonal matrix with its $k$-th diagonal entry given by $ d_k= \frac{1}{\|\mathbf{w}_k\|} $, where $\mathbf{w}_k$ is the $k$-th column of $\mathbf{W}$. Hence, the matrix $\mathbf{D}$ keeps the power allocated to each user unchanged with and without beamforming \cite{Kim05EffeCG}.

Therefore, the received signal vector at the selected $K_{B}$ users is given by
\vspace{-4pt}
\begin{equation}\label{receSigWOsche}
\mathbf{y} = \mathbf{G x} + \mathbf{n} = \mathbf{DPs} + \mathbf{n},
\vspace{-4pt}
\end{equation}
where $\mathbf{n}$ is the additive white Gaussian noise (AWGN) vector with zero mean and variance of $\sigma_n^2$ for each entry.  From (\ref{receSigWOsche}), we get the effective channel gain for user $k$ as given by $d_k^2$.

Under max-min power control, we have the SNR $\xi_k=\xi_j$, $\forall k,j$, where $\xi_k=\frac{P_k d_k^2}{\sigma_n^2}$. Thus, the ergodic per user rate for one resource block is given by
\begin{equation}\label{ergoRate}
R = \mathbb{E}\left\{ \log_2\left( 1+ \frac{P}{\sigma_n^2} \frac{1}{ 1/d_1^2 + \cdots + 1/d_{K_{B}}^2 } \right) \right\}.
\end{equation}

\subsection{Round-robin Scheduler}
The RR scheduler \cite{SUSZF} first assigns all users to multiple groups and then serves all groups in consecutive resource blocks. In the channel magnitude based user grouping, one user's channel magnitude is denoted by the Euclidean norm of its channel vector, which is given by
\vspace{-4pt}
\begin{equation}\label{CMdef}
X_k = \big\|\mathbf{g}_k\big\|^2.
\vspace{-4pt}
\end{equation}
In one scheduling period, the scheduler first constructs one user group by selecting the $K_B$ users with the largest channel magnitudes from the $K$ users, then constructs the following user groups by repeating this way on the residual users. This user grouping method is equivalent to the SUS method \cite{SUSZF} when $M$ is sufficiently large, as the user channels are nearly orthogonal to each other in massive MIMO \cite{Ngo17TWC}.

We use random selection as a performance benchmark for comparison, where the scheduler randomly selects users from the candidate subset for each user group.

\vspace{-4pt}
\subsection{Channel Misreporting Model}
There are two channel magnitude misreporting models: \emph{underreporting} which helps the misreporting user to gain more power unfairly, and \emph{overreporting} which makes the misreporting user transfer some of its power to others. In either model, the users can be grouped in a wrong way which affects the throughput greatly.

In the following sections, we investigate the impact of channel misreporting on two scenarios with homogeneous users and heterogeneous users, respectively, through the evaluation of average rate loss (in percentage) for \true users.

In particular, in the homogeneous case, we assume $\mathbf{B}=\beta \mathbf{I}_K$, and the $K_M$ misreporting users misreport their channel magnitudes by a common scaling ratio, denoted by $\delta$. Without loss of generality, we denote by $\left\{ U_1, \cdots, U_{K_M} \right\}$ the misreporting users' set, by $\left\{ U_{K_M+1}, \cdots, U_{K} \right\}$ the \true users' set. Then the false CSI matrix collected by the BS is given by $\mathbf{F} = \mathbf{\Delta}^{1/2} \mathbf{G}$, where $\mathbf{\Delta}$ is a diagonal matrix with the first $K_M$ diagonal entries given by $\delta$ and the other $K{-}K_M$ diagonal entries given by $1$.

\vspace{-4pt}
\section{Impact of Channel Misreporting on Homogeneous Users}
\vspace{-2pt}
In this section, we quantify the effectiveness of channel magnitude misreporting on homogeneous users with the same average SNR. In this case,  users' channel magnitudes are close to each other due to the channel hardening of massive MIMO. Hence, the user grouping change has a trivial effect on the rate performance (\ref{ergoRate}) of each resource block. As a result, one misreporter can only harm the system by manipulating the power control in the resource block that selects him. Under max-min power control, the misreporting user claims more power by underreporting, hence reducing other users' rates, while it requests less power by overreporting, which benefits others. As we focus on the harm caused by misreporting, in this case, we only investigate the rate losses caused by underreporting.

First, we derive the expression of the rate loss per \true user by considering the case where $M\gg K$ such that all users can be served in one resource block. Then, we move the analysis to the case with arbitrary $K$, where the users are assigned to several resource blocks by RR scheduler.

\vspace{-4pt}
\subsection{All Users Served in One Resource Block}
The rate loss per user is taken in the sense of time average and it is defined as
\vspace{-4pt}
\begin{equation}\label{RLPUfirst}
\theta(K_M|K) = \frac{R^a - R^m}{R^a} = 1 - \frac{R^m}{R^a},
\vspace{-4pt}
\end{equation}
where $K_M|K$ indicates $K_M$ out of $K$ users are misreporters, $R^a$ denotes the ergodic per user rate with all users accurately reporting, and $R^m$ represents the ergodic per user rate regarding \true users with $K_M$ users underreporting.

\subsubsection{User Rate with Accurate Reporting}
It's challenging to derive the exact closed-form expression for the ergodic rate. Similar to \cite{Ngo13ToC}, applying Jensen's inequality, we can obtain a lower bound of the per user rate when $M>K$. This Jensen lower bound becomes exact as $M\rightarrow\infty$ \cite{JensenBoundTightness}.  Hence, the per user rate is expressed as
\vspace{-4pt}
\begin{equation}\label{ergoLowbound}
R^a \rightarrow  \log_2\left( 1 + \frac{P\beta}{\sigma_n^2}\frac{M-K}{K}  \right), ~ \mathrm{as} ~ M \rightarrow \infty.
\vspace{-4pt}
\end{equation}

\subsubsection{User Rate with Underreporting}
In underreporting, the scaling ratio $\delta{<}1$. At the BS side, the scheduler still perceives the false matrix $\mathbf{F}$ as true channels, and performs ZF precoding and power allocation based on  $\mathbf{F}$. The new max-min power control result is denoted by a diagonal matrix $\mathbf{\bar{P}}$. The received signals at the users believed by the BS is shown as
\begin{equation}\label{receSigByBS}
\mathbf{y}^{BS} = \mathbf{F} \cdot \mathbf{F}^H \big( \mathbf{FF}^H \big)^{-1} \mathbf{\bar{D}}  \cdot \mathbf{\bar{P}}\mathbf{s} + \mathbf{n} = \mathbf{\bar{D}}\mathbf{\bar{P}}\mathbf{s} + \mathbf{n},
\end{equation}
where $\mathbf{\bar{D}} =\diag \left\{ \bar{d}_1, \cdots, \bar{d}_K \right\}$. According to Theorem 1 in \cite{Kim05EffeCG}, we have
\begin{equation}\label{effeCGfalseCSI}
\bar{d}_k^2 = \left\| \mathbf{f}_k \mathbf{V}_k \right\|^2,
\end{equation}
where $\mathbf{f}_k$ represents the $k$-th row of the matrix $\mathbf{F}$, and $\mathbf{V}_k \in \mathbb{C}^{M\times (M{-}K{+}1)}$ denotes the orthonormal basis corresponding to the null space of the subspace spanned by the other $K-1$ channels. Note that $\mathbf{V}_k$ is invariant with respect to channel magnitude misreporting, since the channel directions are invariable. Therefore, we obtain
\begin{equation}\label{effeCGFCSIexpr}
\bar{d}_k^2 = \left\{
                \begin{array}{ll}
                  \delta d_k^2, & 1\leq k \leq K_M ; \\
                  d_k^2, & K_M+1 \leq k \leq K .
                \end{array}
              \right.
\end{equation}

The received signals of the users are expressed as
\vspace{-4pt}
\begin{equation}\label{receSigByUE}
\mathbf{y}^{UE} {=} \mathbf{G} {\cdot}  \mathbf{F}^H \big( \mathbf{FF}^H \big)^{-1} \mathbf{\bar{D}}  {\cdot} \mathbf{\bar{P}}\mathbf{s} {+} \mathbf{n}
{=} \big(  \mathbf{\Delta}^{1/2} \big)^{-1}\mathbf{\bar{D}}\mathbf{\bar{P}}\mathbf{s} {+} \mathbf{n}.
\vspace{-4pt}
\end{equation}

From (\ref{receSigByUE}), we observe that $\big(  \mathbf{\Delta}^{1/2} \big)^{-1}\mathbf{\bar{D}} {=} \mathbf{D}$, which means misreporting channel magnitudes does not change users' effective channel gains $\mathbf{D}$. From (\ref{receSigByBS}) and (\ref{receSigByUE}), we observe that for \true users, the SNRs derived by the BS is what they are really receiving, while for misreporting users, the SNR actually received is what the BS calculates multiplied by $1/\delta$.\footnote{Since the powers for misreporting users get larger while the real effective channel gains are invariable, the actual SNRs received by misreporting users are higher than in the case of normally reporting CSI.}

\begin{proposition}\label{pro1}
As $M\rightarrow \infty$,  the per user rate for \true users with $K_M$ users underreporting is given by
\begin{equation}\label{PURfCSIlowB}
R^m {\rightarrow} \log_2\left( 1 + \frac{P\beta}{\sigma_n^2} \frac{M-K}{ K-K_M + \frac{1}{\delta}K_M } \right).
\end{equation}
\end{proposition}
\begin{proof}
The proof is similar to (\ref{ergoLowbound}).
\end{proof}

By substituting (\ref{ergoLowbound}) and (\ref{PURfCSIlowB}) into (\ref{RLPUfirst}), the rate loss when all users are served in one resource block as $M{\rightarrow} \infty$ is given by
\begin{equation}\label{RLPUwoUS}
\theta(K_M|K) {\rightarrow} 1 {-} \frac{\log_2\left( 1 {+} \frac{P\beta}{\sigma_n^2} \frac{M{-}K}{ K{-}K_M {+} \frac{K_M}{\delta} } \right)}{\log_2\left( 1 {+} \frac{P\beta}{\sigma_n^2}\frac{M{-}K}{K}  \right)}.
\end{equation}

Eq. (\ref{RLPUwoUS}) implies that the rate loss $\theta(K_M|K)$ increases with $K_M$, which is intuitive that more misreporting users impose a greater influence.

\textbf{High and Low SNR Regime.}
Through simple derivation, we get that the rate loss decreases with $P$ monotonically. In addition, when $\delta$ is very small, we have $K_M/\delta {\gg} K{-}K_M$. In high SNR regime ($\SNR{\rightarrow} \infty$), we have the rate loss $\theta(K_M|K) {\rightarrow} \log_2{\big(\frac{K_M}{\delta K} \big)} \big/ \log_2{\big( \frac{P\beta}{\sigma_n^2}\frac{M{-}K}{K}  \big)}$, where the numerator denotes the absolute rate loss. It is seen that the absolute rate loss is invariant with respect to the antenna number, the BS power, and the large-scale fading. However, the rate loss in percentage decreases with these parameters.
In low $\SNR$ regime ($\SNR{\rightarrow} {-\infty}$)
, we have $\theta(K_M|K) {\rightarrow} 1 - \frac{\delta K}{K_M}$  by using the taylor series expansion, which implies that the rate loss is very large when $\SNR$ is very small. For example, when $\delta{=}-20$dB, the rate loss for \true users is about 68\% as one out of 32 users underreports its channel.

\vspace{-4pt}
\subsection{Schedule Users by Round-robin Scheduler}
The RR scheduler serves all users once in one scheduling period.
For brevity, we only consider $K_M\leq K_{B}$ to unveil valuable insights of the impact and see the results with more misreporting users via simulations. All misreporting users underreports CSI using the same scaling-down ratio $\delta$. Besides, $\delta$ is assumed to be very small, in order to show the limit of the damage that CSI underreporting can bring to the system.

\textbf{Scheduling of the Misreporting Users.}
Based on RR scheduler with channel magnitude based user grouping, when $\delta$ is very small, the misreporting users are almost surely scheduled to the last resource block of each scheduling period.
The explanation is given below.

It is easily shown that the \true user's channel magnitude $X_k$ obeys the gamma distribution $\Gamma(M,\beta)$ with shape-scale parameters ($M,\beta$). For misreporting users, the reported channel magnitude is $X_i^m {=} \delta \| \mathbf{g}_i \|^2$, hence we have $X_i^m \sim \Gamma(M,\delta\beta)$. Then we can have the probability $\mathrm{Pr}(X_i^m < X_k) \rightarrow 1$, $\forall i,k$, as $\delta\rightarrow 0$, where the convergence follows from the fact that the mean $\mathbb{E}[X_i^m]=M\delta\beta \ll \mathbb{E}[X_k]$ and the variance $\mathbb{V}\mathrm{ar}[X_i^m]{=}M\delta^2\beta^2 {\ll} \mathbb{V}\mathrm{ar}[X_k]$, as $\delta\rightarrow 0$. Therefore, the misreporters are grouped together and scheduled to the last resource block, as their channel magnitudes are lower than those of \true users.

\textbf{Scheduling of the \True Users.}
Like the normal scenario, the RR scheduler groups the residual $K-K_M$ \true users, by exploiting their multi-user diversities.

The time-average rate loss for \true users based on RR scheduler with channel magnitude based user selection, is defined as
\begin{equation}\label{rlpuCM}
\theta^{CM} {=} \frac{ R^{a,CM}{\big/}T  {-}R^{m,CM}{\big/}T }{ R^{a,CM}{\big/}T } {=} \frac{ R^{a,CM}  {-}R^{m,CM} }{ R^{a,CM} },
\end{equation}
where $R^{a,CM}$ denotes the per user rate across one scheduling period with accurate CSI reporting, and $R^{m,CM}$ is the per user rate for \true users across one scheduling period with CSI misreporting. Denote by $R_t^{a,CM}$ and $R_t^{m,CM}$ the per user rate in terms of the $t$-th resource block with accurate reporting and misreporting, respectively. Then $R^{a,CM}$ and $R^{m,CM}$ are separately given by
\vspace{-4pt}
\begin{gather}
R^{a,CM} = \frac{1}{K} \sum_{t=1}^{T} K_{B} R_{t}^{a,CM}.  \label{PURscheTCSI}\\
R^{m,CM} {=} \frac{ (K_{B}{-}K_M)R_{T}^{m,CM} {+}  \sum_{t=1}^{T{-}1} K_{B} R_{t}^{m,CM} }{K-K_M}. \label{PURscheFCSI}
\vspace{-4pt}
\end{gather}

Similar to (\ref{PURscheTCSI}), regarding the random user grouping under accurate CSI reporting, define $R^{a,rand}$ and $R_t^{a,rand}$ as the per user rate for $K$ users across one scheduling period and per user rate with respect to the $t$-th resource block, respectively, In massive MIMO, we have the following proposition.

\begin{proposition}\label{massiveMprop}
As $M\rightarrow \infty$, the following assertions are true:

1. Random user grouping performs almost the same as the channel magnitude based user grouping under accurate CSI reporting, i.e.,
\begin{equation}
R_{t}^{a,CM} - R_{t}^{a,rand} \rightarrow 0, ~~ 1 \leq t \leq T. \label{CMrandTindt}
\end{equation}
And it is easily shown that $R^{a,rand} = R_{t}^{a,rand}$, $\forall t$.

2. With $K_M$ ($K_M{<}K_B$) users underreporting, for the channel magnitude based user grouping, the rate losses in the preceding $T{-}1$ resource blocks are very small and can be ignored, i.e.,
\begin{equation}\label{CMTorFrates}
R_{t}^{m,CM} - R_{t}^{a,CM} \rightarrow 0, ~~ 1 \leq t \leq T-1.
\end{equation}
\end{proposition}

\begin{remark}
The convergences in Proposition \ref{massiveMprop} are essentially due to the favorable propagation and channel hardening of massive MIMO \cite{Ngo17TWC}.  Thus the benefit of channel magnitude based user selection from multi-user diversity \cite{cLi17TWC} over random selection is trivial, hence we have (\ref{CMrandTindt}). Besides, in underreporting, the rate losses in the preceding $T{-}1$ resource blocks come from the multi-user diversity loss, since the \true users' number decreases to $K{-}K_M$. However, this loss is also negligible due to channel hardening. Thus we have (\ref{CMTorFrates}). Note that this proposition is not valid for heterogeneous users.
\end{remark}

Then, we can derive the rate loss of channel magnitude based user grouping as shown in the following proposition.

\begin{proposition}\label{propoRLPUCM}
As $M{\rightarrow}\infty$,
\begin{equation}\label{RLPUscheCM}
\theta^{CM}
{\rightarrow} -\frac{K_M(T{-}1)}{T(K{-}K_M)} {+} \frac{ R_{T}^{a,CM} }{ T R^{a,rand} } {-} \frac{ (K_{B}{-}K_M)R_{T}^{m,CM} }{(K{-}K_M)R^{a,rand}},
\end{equation}
in which the asymptotic expressions of $R^{a,rand}$, $R_{T}^{a,CM}$ and $R_{T}^{m,CM}$ are given by, respectively
\begin{gather}
R^{a,rand}  \rightarrow  \log_2\left( 1 + \frac{P\beta}{\sigma_n^2} \frac{ M-K_{B} }{ K_{B} } \right), ~ \mathrm{as} ~ M\rightarrow\infty, \label{PURTCSIrand} \\
R_{T}^{a,CM} {\rightarrow} \log_2\left( 1 + \frac{P}{\sigma_n^2} \frac{ M-K_{B} }{ (M-1) A_{T}^a }  \right), ~ \mathrm{as} ~ M\rightarrow\infty, \label{RLPUlastTSlowBTCSI} \\
R_{T}^{m,CM} {\rightarrow}   \log_2\left( 1 + \frac{P}{\sigma_n^2} \frac{ M-K_{B} }{ (M-1) A_{T}^m }  \right), ~ \mathrm{as} ~ M\rightarrow\infty, \label{RLPUlastTSlowBFCSI}
\end{gather}
where $A_{T}^a = \sum_{k=1}^{K_{B}} \int_0^{+\infty} \frac{ f_{(k)}(x) }{x} dx $, and $A_{T}^m = \frac{K_M}{\delta\beta(M-1)} + \sum_{k=1}^{K_{B} -K_M } \int_0^{+\infty} \frac{ \bar{f}_{(k)}(x) }{x} dx $, and $f_{(k)}(x)$ and $\bar{f}_{(k)}(x)$ represent the probability density functions (PDFs) of the $k$-th smallest order statistic of the order statistics for a sample of size $K$ and $K-K_M$, respectively, both from the gamma distribution $\Gamma(M,\beta)$ with shape-scale parameters $(M,\beta)$.
\end{proposition}
\begin{proof}
Please see Appendix-A in \cite{supplementFile}.
\end{proof}

However, it's difficult to observe valuable insights from (\ref{RLPUscheCM}). Next, we provide an upper bound of this rate loss.

\begin{corollary}\label{upperBoundOfRateLoss}
The upper bound of the rate loss per honest user with $K_M$ users underreporting is given by
\begin{equation}\label{thetaCMupperB}
\theta_{upper}^{CM} = \frac{ K_{B} - K_M }{ K - K_M } \theta(K_M|K_{B}), ~~ K_M \leq K_{B},
\end{equation}
where $\theta(K_M|K_{B})$ is determined by (\ref{RLPUwoUS}).
\end{corollary}
\begin{proof}
Please see Appendix-B in \cite{supplementFile}.
\end{proof}

Moreover, when $K_M = 1$, it is easily shown that $\theta_{upper}^{CM}=\theta^{rand}$, where $\theta^{rand}$ represents the rate loss per \true user for the RR scheduler that uses random user grouping.

\begin{figure}
\setlength{\abovecaptionskip}{-1mm}
\centering
\includegraphics[width=2.7in]{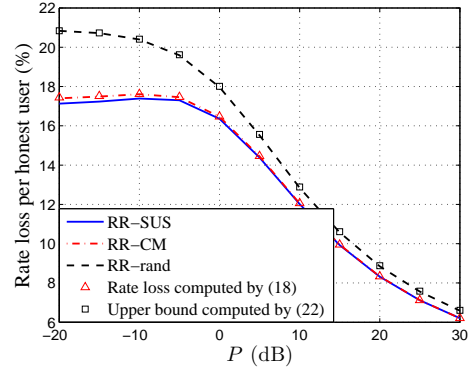}
\caption{Rate loss per \true user vs. $P$ with one underreporting user.}\label{RLPUoneAttaHomodiffP}
\vspace{-7mm}
\end{figure}

From (\ref{thetaCMupperB}), we see that the rate loss is zero when $K_M{=}K_{B}$, indicating no harm with $K_{B}$ users' underreporting. The reason is that $K_B$ misreporting users occupy the last resource block and no \true user is grouped together with any of them. In addition, (\ref{thetaCMupperB}) shows that the impact of misreporters on the channel magnitude based user grouping is partly alleviated through the coefficient $\frac{K_{B} -K_M }{ K-K_M }$, as the misreporters are grouped together in the last resource block and discontinuously selected by the scheduler. 

\vspace{-6pt}

\subsection{Performance Evaluation}\label{numeResHom}
We set $M{=}64$ for massive MIMO implementation. The noise variance and the large-scale fading are set to be $\sigma_n^2{=}1$ and $\beta=1$, respectively. We set $K{=}32$, $K_{B}{=}8$. The underreporting ratio $\delta$ is set to be $\delta{=}-20$dB.

Fig. \ref{RLPUoneAttaHomodiffP} depicts the rate loss for \true users versus the power $P$ at the BS, considering one misreporting user. The legend ``RR-\{SUS, CM, rand\}'' indicates the simulation results of RR scheduler with SUS, channel magnitude based and random user grouping, respectively. Besides, the analytic results of the rate loss in (\ref{RLPUscheCM}) and the upper bound in (\ref{thetaCMupperB}) are also plotted. This figure proves the good match between the analytic and simulated results, and verifies that (\ref{thetaCMupperB}) represents the result for RR scheduler with random selection when there is one misreporting user. In addition, it is observed that with massive antennas, the SUS algorithm can be well approximated by the channel magnitude based user selection, and the \true users in the case with random selection suffer higher losses. Moreover, Fig. \ref{RLPUoneAttaHomodiffP} verifies that the rate loss is large when the power is small and decays as the power increases.

In Fig. \ref{rlpuVsKaHomodifKa}, for the channel magnitude based RR scheduler, the analytic result (\ref{RLPUscheCM}) and its upper bound (\ref{thetaCMupperB}) are verified to be tight when the system has multiple misreporters. Besides, we observe an interesting phenomenon that the rate loss shows the periodic increase and decrease as $K_M$ grows, with $K_B$ as the period. In contrast, the random user selection is verified to be much more sensitive to more users' underreporting.

\begin{figure}
\setlength{\abovecaptionskip}{-1mm}
\centering
\includegraphics[width=2.7in]{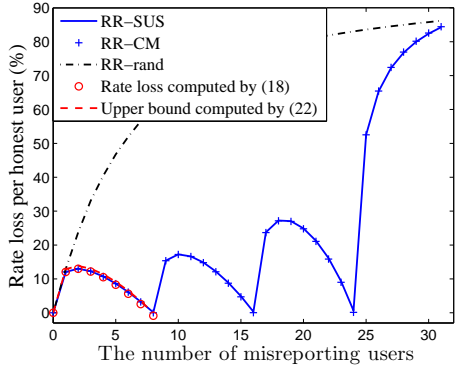}
\caption{Rate loss per \true user vs. $K_M$, $P=10$dB.}\label{rlpuVsKaHomodifKa}
\vspace{-7mm}
\end{figure}

\textbf{Periodicity Analysis.}
The ``periodic'' phenomenon of the rate loss function in Fig. \ref{rlpuVsKaHomodifKa} can be understood by analyzing (\ref{thetaCMupperB}).
In (\ref{thetaCMupperB}), when $K_M {\leq} K_{B}$, we know from (\ref{RLPUwoUS}) that $\theta(K_M|K_{B})$ increases with $K_M$, while it is easily shown that $\frac{K_{B} -K_M }{ K-K_M }$ decreases with $K_M$. Therefore, it is inferred that $\theta_{upper}^{CM}$ first increases with $K_M$ when the growing speed of $\theta(K_M|K_{B})$ is larger than the descending speed of $\frac{K_{B} -K_M }{ K-K_M }$, and then decreases with $K_M$ due to the opposite reason. In addition, from (\ref{thetaCMupperB}), we can predict the tendency of the rate loss when $K_M$ is larger than $K_{B}$. For example, when $K_{B}{\leq} K_M {\leq} 2K_{B}$, we will lose the last resource block that is full of misreporters. And the problem will reduce to that using $K_M{-}K_{B}$ misreporting users out of $K{-}K_{B}$ total users, and hence the similar result like (\ref{thetaCMupperB}) will be acquired. As a result, the rate loss per \true user will first increase and then decrease to zero, periodically with $K_M$, and the period is equal to $K_{B}$. And when $K_M{\geq} (T-1)K_{B}$, the rate loss will perform like $\theta(K_M|K_{B})$.

\vspace{-3pt}

\section{Impact of Channel Misreporting on Heterogeneous Users}
In this section, we consider the heterogeneous users case with various average SNRs, which is more practical in the real world. Due to the advantage of channel hardening, the scheduler in this case can depend only on the large-scale fading in time domain \cite{Ngo17TWC}. For simplicity, we mainly analyze the scheduler depending on the large-scale fading coefficients, and compare via simulations with the scheduler based on the channel magnitudes that vary once the small-scale fading changes. In addition, we assume the worst scenario that the misreporters have \true users' average SNRs, and only modify their own large-scale fading. In this scenario, our objective is to find out the most threatening misreporting strategies.

The users are relabeled based on their large-scale fading so that we have $\beta_1{>}\beta_2{>} \cdots {>} \beta_K$. Considering one RR scheduling period, one example of the scheduling result with accurate CSI reporting is shown in Fig. \ref{userGroupingChange}-(a). For the $t$-th resource block, similar to (\ref{ergoLowbound}), the asymptotic expression of the ergodic per user rate under max-min power control as $M\rightarrow \infty$ is given by
\vspace{0pt}
\begin{equation}\label{PURLSKzfTCSI}
R_{LS,t}^a {\rightarrow} \log_2{ \Big( 1 {+} \frac{P}{\sigma_n^2} \frac{M{-}K_{B}}{ 1{/}\beta_{t,1} {+} {\cdots} {+} 1{/}\beta_{t,K_{B}} }  \Big) }, ~ 1{\leq} t {\leq} T.
\vspace{0pt}
\end{equation}
where $\beta_{t,k}=\beta_{(t-1)K_B+k}$. In (\ref{PURLSKzfTCSI}), it is seen that one resource block has different rate performance by selecting different users, as their large-scale fading coefficients differ greatly.

\begin{figure}
\setlength{\abovecaptionskip}{-1mm}
  \centering
  \includegraphics[width=3.4in]{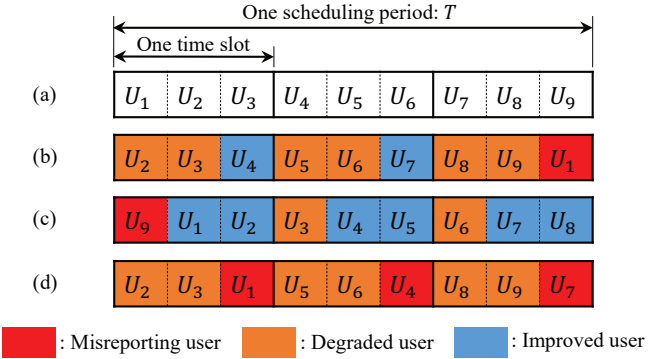}
  \caption{Scheduling results in the heterogeneous case with and without misreporting, respectively. $K=9$, $K_B=3$. (a) With accurate CSI reporting. (b) In underreporting, $U_1$ underreports by $\beta_1^m<\beta_9$. (c) In overreporting, $U_9$ overreports by $\beta_9^m>\beta_1$. (d) In user grouping unchanged underreporting, $U_1$ underreports by $\beta_4$, $U_4$ underreports by $\beta_7$, and $U_7$ underreports by $\beta_7^m<\beta_9$.}\label{userGroupingChange}
  \vspace{-7mm}
\end{figure}
\textbf{How and How Much Should the Underreporting Be?}
From (\ref{PURLSKzfTCSI}), we see that for one resource block, underreporting the best user's large-scale fading to be lower than the worst user's can reduce the rate to the maximum extent. However, under RR scheduler, the user grouping result will change if the underreporting is too low: the misreporter will be delayed to one future resource block of the same scheduling period and replaced by one user in the next resource block, like Fig. \ref{userGroupingChange}-(b) shows where $U_1$ underreports by $\beta_1^m<\beta_9$. Indeed, letting the best user $U_1$ underreport to the lowest is the worst case for the degraded users. Fig. \ref{userGroupingChange}-(b) shows that the rate loss comes from two parts: the unfair power allocation in the last resource block that holds the misreporter and the user grouping change in the preceding $T{-}1$ resource blocks.

\textbf{Can Overreporting Harm the System Performance?}
From (\ref{PURLSKzfTCSI}), we observe that in one resource block, if one user overreports, the other users in that resource block will get a higher rate. Similar to the underreporting, however, the scheduling result will change if the overreporting is too high: the misreporter will be scheduled to one preceding resource block and replaced by one user in the previous resource block, like in Fig. \ref{userGroupingChange}-(c) where user $U_9$ overreports by $\beta_9^m>\beta_1$. Hence, the overreporting can also decrease some users' rates by changing the user grouping. Note that the rate loss in overreporting only comes from the user grouping change, and in the resource block that has the misreporter, the \true users get rate improvements.

Therefore, both underreporting and overreporting are able to change the scheduling result and decrease some users' rates. Note that the user grouping change also helps to improve some users' rates by grouping them with better users, like the user $U_4$ in Fig. \ref{userGroupingChange}-(b) and Fig. \ref{userGroupingChange}-(c).

\vspace{-4pt}

\subsection{Misreporting Strategy Changing User Grouping}\label{misStrChangingUG}
\vspace{-2pt}
Based on (\ref{PURLSKzfTCSI}) and Fig. \ref{userGroupingChange}, we obtain that the user grouping change is the worst case for the degraded users. Accordingly, we propose a misreporting strategy that aims to harm the degraded users to the most extent, regardless of some gaining benefits.
The misreporting strategy is described as follows.

\textbf{Underreporting that Totally Changes User Grouping.}
\textit{In underreporting, choose the $K_M$ users with the largest average SNRs and underreport their large-scale fading to the lowest value $\beta_{low}^{m}$ ($\beta_{low}^{m} < \beta_{K}$)}.

\textbf{Overreporting that Totally Changes User Grouping.}
\textit{In overreporting, choose the $K_M$ users with the lowest average SNRs and overreport their large-scale fading to the highest value $\beta_{high}^{m}$ ($\beta_{high}^{m} > \beta_{1}$)}.

\textbf{Periodicity Property.}
For this type of misreporting, the adversary would rather choose $K_M{<} K_{B}$ than use more misreporters. It's because when $K_M{=}K_{B}$, the misreporting users would occupy one resource block and the user grouping for \true users is invariable, thus no \true user is affected. In addition, when $K_M{>}K_{B}$, the remaining \true users would have the same rate losses with the corresponding users in the case that has $K_M{-}K_{B}$ misreporters, which means the rate loss of one specific \true user has a periodic property as $K_M$ grows. For example, in Fig. \ref{userGroupingChange}-(b), user $U_7$ has the same rate loss between the case where one user $U_1$ with the largest SNR underreports by $\beta_{low}^{m}$ and the case where the best four users $U_1$ to $U_4$ underreport to the lowest $\beta_{low}^{m}$.

Note that the user grouping result with $K_M$ users underreporting is the same as that using $K_{B}{-}K_M$ users overreporting.

\vspace{-0pt}
\subsection{Misreporting Strategy Keeping User Grouping Unchanged}
Fig. \ref{userGroupingChange} shows that there are users getting higher rates as long as the user grouping changes. In this subsection, we investigate the case where the misreporting degrees are not so large such that the user grouping result is not changed. To keep the user grouping fixed, the misreporting degrees have to be carefully designed. Besides, since the user grouping is unchanged, the overreporting cannot harm other users, instead,  it improves other users' rates based on (\ref{PURLSKzfTCSI}). Hence, we only analyze the case of underreporting.

Additionally, the misreporters aim to increase the average rate loss on \true users as the number of misreporters $K_M$ grows. Hence, the main idea is to reduce the per user rate in each resource block as low as possible. In particular, based on (\ref{PURLSKzfTCSI}), in one resource block, we choose the current best user as the misreporter and underreport its large-scale fading to be lower than that of the current worst user in that resource block, but higher than the large-scale fading coefficients of the next resource block to keep the user grouping unchanged.

Next, considering one scheduling period, as $K_M$ increases, we provide an algorithm that shows how to choose the misreporting users and underreport their channel magnitudes.

\vspace{2pt}
\textbf{User Grouping Unchanged Underreporting Algorithm.}
\vspace{0.5mm}
\hrule
\vspace{0.5mm}
Step 1) Initialization. Set $K_M=1$. Denote by $t_m$ the index of the resource block where the new misreporter is, and set $t_m=1$. Denote by $I_t$ the set of the misreporters' indices in the $t$-th resource block, and set $I_t=\emptyset$ for $1\leq t\leq T$.

Step 2) Select the user $U_{i_m}$ as the misreporter. The index of this new misreporter is determined by
\begin{equation}\label{idxNewMisreporter}
i_m = (t_m-1)K_{B} + \Big\lceil \frac{K_M}{T} \Big\rceil,
\end{equation}
where $\lceil x\rceil$ is the ceiling function of $x$.
Then update the set of misreporters' indices as $I_{t_m} {\leftarrow} [I_{t_m}; i_m]$.

Step 3) When $1{\leq} t_m{<}T$. If $I_{t_m{+}1}{\neq} \emptyset$, let the misreporters in the $t_m$-th resource block underreport by $\beta_{I_{t_m}}^m {=} \beta_{I_{t_m{+}1}(end)}$; else, let the misreporter $U_{i_m}$ underreport by $\beta_{i_m}^m{=}\frac{\beta_{t_mK_{B}}{+}\beta_{t_mK_{B}{+}1}}{2}$. Then update the underreporting degree of the misreporters in the previous resource block by $\beta_{I_{t_m-1}}^m\leftarrow \beta_{i_m}$ if $t_m>1$.

When $t_m=T$. Let the misreporters in the last resource block underreport by $\beta_{I_{T}}^m {=} \beta_{low}^{m}$. Then update $\beta_{I_{T{-}1}}^m {\leftarrow} \beta_{i_m}$.

Step 4) Update $K_M{\leftarrow} K_M{+}1$. Then the next misreporter locates in the resource block with index updated by $t_m{=} \mod{(\frac{K_M}{T})}$, and change $t_m$ to $T$ if $t_m{=}0$. Then go to step 2.
\vspace{0mm}
\hrule
\vspace{1mm}

An example of this algorithm is shown in Fig. \ref{userGroupingChange}-(d).

\textbf{Misreporting Effect Analysis.}
Under this misreporting strategy, no user gets rate improvement, as the user grouping result keeps fixed. Unlike the misreporting strategy in Section \ref{misStrChangingUG} where multiple misreporters are grouped together, this misreporting strategy distributes multiple misreporters in different resource blocks. As the user rate in each resource block that has misreporters is reduced as low as possible, the rate loss on \true users is predicted to increase as the number of misreporters increases, which is different from the ``periodicity'' property of the misreporting strategy that changes the user grouping.

\subsection{Performance Evaluation}
Taking into account the path loss and shadow fading, the large-scale fading coefficient $\beta_k$ is expressed as
\vspace{-0pt}
\begin{equation}\label{LScoeffi}
\beta_k = \frac{10^{\omega_k/10}}{ 1 + \left( d_k/d_0 \right)^l },
\vspace{-0pt}
\end{equation}
where $10^{\omega_k/10}$ represents the shadow fading in log-normal distribution with standard derivation of $\sigma$ dB, and $\omega_k \sim \mathcal{N}(0,\sigma^2)$ which is the normal distribution expressed in dB; and $d_k$ denotes the distance between $U_k$ and the BS, which is uniformly distributed between 0 and the cell radius $r$, and $d_0$ represents a reference distance; $l$ shows the path loss exponent.

As in Section \ref{numeResHom}, we set $M=64$, $K=32$, $K_{B}=8$, and $P=10$dB. We consider the cell radius $r = 500$m, and the path loss and shadow fading parameters are set to be: $l=3.8$, $d_0=200$m, and $\sigma = 8$dB. The lowest underreporting and the highest overreporting values are assumed to be $\beta_{low}^{m} {=} \beta_K/2$ and $\beta_{high}^{m} {=} 2\beta_1$, respectively.
We use the markers ``RR-LS'' and ``RR-CM'' to denote the RR scheduling depending only on the large-scale fading and the scheduling depending on the channel magnitudes which vary once the small-scale fading changes, respectively.

In Fig. \ref{RLindKa4underR}, we plot the rate change of every \true user with channel underreporting that totally changes user grouping, i.e., four users $U_1$ to $U_{4}$ underreport by $\beta_{low}^{m}$. This figure verifies that user grouping change can greatly decrease some \true users' rates in heterogeneous users case. For example, the degraded users due to user grouping change see higher than 18\% rate losses, that are users $U_5$ to $U_8$, $U_{13}$ to $U_{16}$, and $U_{21}$ to $U_{24}$. In the last resource block that holds the four misreporters, users $U_{29}$ to $U_{32}$ suffer about 68\% rate losses due to the unfair power allocation to the misreporters. Besides, it is verified that some users benefit from user grouping change as they are grouped with better users and grab some power from them.

\begin{figure}
\setlength{\abovecaptionskip}{0mm}
\centering
\includegraphics[width=2.7in]{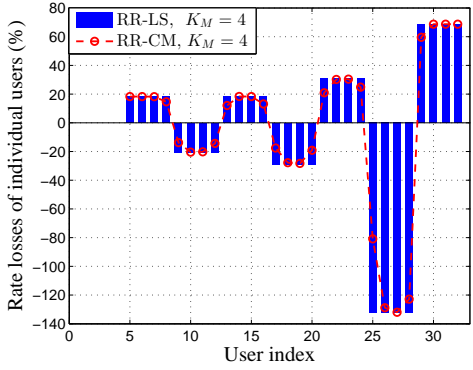}
\caption{Rate losses of individual users in user grouping changed underreporting. $U_{1}$ to $U_{4}$ underreports by $\beta_{low}^{m}$.}\label{RLindKa4underR}
\vspace{-4mm}
\end{figure}

\begin{figure}
\setlength{\abovecaptionskip}{0mm}
\centering
\includegraphics[width=2.7in]{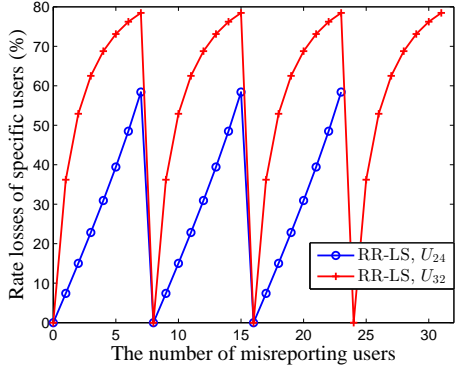}
\caption{Rate losses of users $U_{24}$ and $U_{32}$ vs. $K_M$ in user grouping changed underreporting.}\label{RLspeHetdiffKA}
\vspace{-5mm}
\end{figure}

Moreover, Fig. \ref{RLindKa4underR} shows that the differences of the rate losses between scheduling on large-scale fading and scheduling on channel magnitudes are small.

Fig. \ref{RLspeHetdiffKA} shows the rate losses of users $U_{24}$ and $U_{32}$ versus the number of misreporters $K_M$ in underreporting that totally change the user grouping. The figure demonstrates that one specific \true user has a periodic property as increasing $K_M$, suggesting the adversary to choose less than $K_{B}$ users to underreport as using this misreporting strategy.

Fig. \ref{RLPUHeterallTSdiffKA} describes the rate loss regarding the \true users with user grouping unchanged underreporting. Besides, rate losses of the case with user grouping changed underreporting are also plotted. By comparison, we see that the rate loss under user grouping unchanged underreporting increases with $K_M$ almost linearly, rather than periodically as in user grouping changed underreporting. For example, ten users underreporting can make other 22 \true users suffer about 18\% rate loss on average. Moreover, the random selection is observed to be more vulnerable to underreporting under both misreporting strategies, when $K_M$ is large.

In Fig. \ref{HetallTSdiffKAdifTSKMode2}, we investigate the impacts of the scheduling period $T$ and user number per resource block $K_{B}$ on the effectiveness of misreporting. The user grouping unchanged underreporting is adopted and the rate loss on  \true users is plotted. Fig. \ref{HetallTSdiffKAdifTSKMode2} shows that when the total user number $K$ is fixed, serving fewer users in one resource block can greatly mitigate the misreporting effect. However, in this case, the total system throughput loss is larger, as less multiplexing gain is utilized. Besides, as the scheduling period $T$ is fixed, it needs more misreporters to achieve the same amount of rate loss when one resource block supports more users.

\begin{figure}
\setlength{\abovecaptionskip}{0mm}
\centering
\includegraphics[width=2.7in]{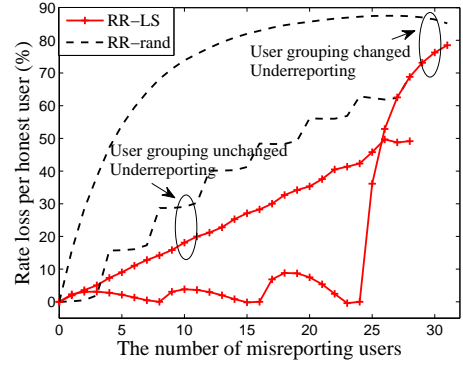}
\caption{Rate loss per \true user vs. $K_M$ under both user grouping unchanged underreporting and user grouping changed underreporting, respectively.}\label{RLPUHeterallTSdiffKA}
\vspace{-4mm}
\end{figure}

\begin{figure}
\setlength{\abovecaptionskip}{0mm}
\centering
\includegraphics[width=2.7in]{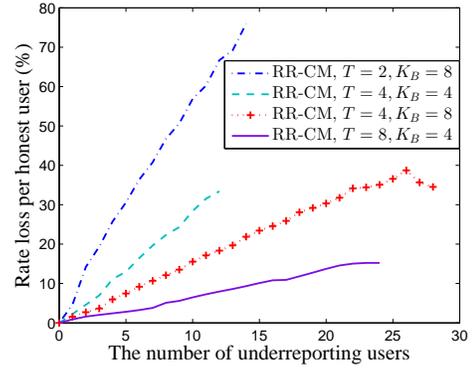}
\caption{Rate loss per \true user vs. $K_M$ in user grouping unchanged underreporting.}\label{HetallTSdiffKAdifTSKMode2}
\vspace{-5mm}
\end{figure}

\section{Related Work}
To the best of our knowledge, the impact of CSI misreporting on both the scheduler and precoding has not been previously studied. Work related to ours can be categorized into serving all users in one time-frequency resource block and scheduling with one user in one resource block.

\textbf{Serving All Users in One Time-Frequency Resource Block.}
The vulnerabilities of multi-user MIMO systems by reporting forged CSI was primitively investigated in \cite{Muk10ICASSP}, which considered the downlink of a multi-user multicast network that transmits a common message to multiple receivers. In \cite{Muk10ICASSP}, the authors mainly analyzed the influence on Quality-of-Service (QoS) for \true receivers caused by various forms of malicious channel feedback, and showed that the presence of just a single adversary could degrade the system performance significantly. Then, for the purposes of eavesdropping and selfish, the impact of channel state misreporting of a malicious user was studied in \cite{tung2014vulnerability} for a multi-user MIMO system that broadcasts independent messages to multiple users. The proposed ``sniffing attack" is an example of the channel direction modified CSI misreporting, which enables the attacker to eavesdrop the signal sent to the legitimate user. Besides, the ``power attack'' proposed in \cite{tung2014vulnerability} is to manipulate the Access Point's power allocation. By modifying the reported channel magnitude, the adversary can enhance its own capacity at the expense of others'. These purely physical layer analyses show the great harm of channel misreporting. Our analysis, however, shows that user scheduling can alleviate the damage and the rate loss shows a different characteristic of ``period'' increase and decrease with increasing number of misreporters.

\textbf{Scheduling with One User in One Time-Frequency Resource Block.} In the medium access control (MAC) layer,
there are some papers investigating the impact of the cheating of channel conditions on the scheduler performance \cite{Racic10TMC, kim2014falseCSIreport}. They considered the current 3G and LTE systems which have not deployed multi-user MIMO.  The scheduler selects only one user in each time-frequency resource block. The main idea to be malicious in this case is to steal as many resource blocks as possible \cite{kim2014falseCSIreport}. Hence, overreporting is a good attacking strategy for both proportional fair (PF) and max throughput schedulers, which try to select a user with a relatively good channel quality at each resource block. This implies underreporting inflicts no damage on others. For the round-robin scheduler, neither of overreporting and underreporting can gain resource blocks unfairly, as round-robin doesn't rely on channel conditions. In multi-user systems, however, multiple users share the transmission power in one joint transmission at each resource block. Therefore, both overreporting and underreporting have a chance to inflict threats to the system,  by subverting the user grouping result of the scheduler.

\section{Concluding Discussions}
We identified the threats of channel state misreporting on multi-user scheduling performance in massive MIMO systems, and showed that the performance of the fair round-robin scheduler combined with max-min fairness power control is very sensitive to channel magnitude misreporting. The rate losses of scheduling performance consist of two parts: (i) the rate loss of honest users in the resource block that holds the misreporting users and (ii) the rate loss caused by the change of user grouping. In addition, the ``periodic'' property of the rate loss function with the number of misreporters verifies how the multi-user scheduler performance loss analysis differs from the purely physical layer analyses and the analyses on the current scheduling systems. Moreover, we proposed two types of efficient misreporting strategies on the practical heterogeneous users case. Finally, numerical results verify the accuracy of the analytical results and show the effectiveness of the proposed misreporting strategies.


For FDD massive MIMO, it is still an open area for study in both literature and practical deployment. The most difficulty in FDD massive MIMO is to reduce the large overhead in CSI acquisition. Recently, \cite{xing15AoA} showed that it's possible to greatly reduce the overhead by exploiting the reciprocity of dominant AoA in massive MIMO. Additionally, a large spectrum of the current networks operate in the FDD mode, hence, the upgrade of the current FDD network to massive MIMO is highly possible. Actually, misreporting can be a problem in any system, FDD or time-division duplex (TDD). In the uplink channel estimation phase of TDD systems, the misreporter can transmit at a forged transmit power, hence misleading the BS to estimate a false channel magnitude for the misreporter. Therefore, the impact of misreporting channel magnitudes in FDD can be straightforwardly applied to TDD systems.

\vspace{0pt}

\vspace{5mm}
\bibliography{reference}
\end{document}